\documentclass[11pt]{article}

\usepackage{amsmath,amssymb,amsthm}
\usepackage{algorithm,algorithmic}
\usepackage{listings}
\usepackage{xcolor}
\usepackage{hyperref}
\usepackage{tikz}

\lstdefinestyle{pythonstyle}{
    language=Python,
    basicstyle=\ttfamily\small,
    keywordstyle=\color{blue}\bfseries,
    commentstyle=\color{gray}\itshape,
    stringstyle=\color{green!60!black},
    showstringspaces=false,
    breaklines=true,
    frame=single,
    numbers=left,
    numberstyle=\tiny\color{gray}
}

\newtheorem{definition}{Definition}
\newtheorem{theorem}{Theorem}

\newtheorem{lemma}{Lemma}
\newtheorem{observation}{Observation}

\title{The Beautiful Deception:\\
How 256 Bits Pretend to be Infinity}

\author{
    Alexander Towell\\
    \texttt{Southern Illinois University Edwardsville}\\
    \texttt{Southern Illinois University Carbondale}\\
    \texttt{atowell@siue.edu, lex@metafunctor.com}
}

\date{\today}

\begin{document}

\maketitle

\begin{abstract}
How do you store infinity in 256 bits? This paper explores the fundamental deception at the heart of computational cryptography: using finite information to simulate infinite randomness. We prove why true random oracles are impossible, then show how lazy evaluation creates a beautiful lie---a finite automaton that successfully pretends to be infinite. We reveal that ``randomness'' in cryptography is actually computational hardness in disguise, demonstrating through Python implementations how 256 bits of entropy can generate sequences indistinguishable from infinite randomness to any computationally bounded observer.
\end{abstract}

\section{Introduction: The Impossible Oracle}

Consider this paradox: cryptographic theory routinely assumes the existence of random oracles~\cite{bellare1993random}---functions that map any input to an infinite sequence of perfectly random bits. Yet here we are, working with computers that can barely store a few terabytes. How can we possibly implement something that requires infinite storage?

The answer is: we can't. And we don't. Instead, we engage in an elaborate and beautiful deception.

This paper tells the story of that deception. We begin by attempting to implement a true random oracle---a class we'll call \texttt{OracleDigest}---and watch it fail in every possible way: running out of memory, unable to persist, impossible to share. This failure is not a bug; it's a constructive proof that true random oracles cannot exist in our computational universe.

But from this impossibility emerges something remarkable: a deterministic construction we'll call \texttt{LazyDigest} that uses just 256 bits of entropy to generate what appears to be an infinite random sequence. The sequence isn't truly random---it has low Kolmogorov complexity and must eventually cycle. Yet to any computationally bounded observer, it's indistinguishable from true randomness.

This is the beautiful deception: we're not generating randomness at all. We're generating computational hardness and calling it randomness.

\section{Mathematical Foundations}

Before we can appreciate the deception, we need to understand what we're trying to fake. This section provides the essential cryptographic definitions that underpin our exploration.

\subsection{Hash Functions and Random Oracles}

\begin{definition}[Probabilistic Polynomial Time (PPT)]
An algorithm $\mathcal{A}$ runs in probabilistic polynomial time if there exists a polynomial $p(n)$ such that for all inputs of length $n$, $\mathcal{A}$ terminates within $p(n)$ steps and may use randomness in its computation.
\end{definition}

\begin{definition}[Cryptographic Hash Function]
A function $h: \{0,1\}^* \to \{0,1\}^n$ is a cryptographic hash function if it satisfies:
\begin{enumerate}
\item \textbf{Pre-image resistance}: For randomly chosen $y \in \{0,1\}^n$,
$$\Pr[\mathcal{A}(y) \to x : h(x) = y] \leq \text{negl}(n)$$

\item \textbf{Second pre-image resistance}: For any $x_1 \in \{0,1\}^*$,
$$\Pr[\mathcal{A}(x_1) \to x_2 : x_2 \neq x_1 \land h(x_1) = h(x_2)] \leq \text{negl}(n)$$

\item \textbf{Collision resistance}: 
$$\Pr[\mathcal{A}() \to (x_1, x_2) : x_1 \neq x_2 \land h(x_1) = h(x_2)] \leq \text{negl}(n)$$
\end{enumerate}
where $\mathcal{A}$ is any PPT adversary and $\text{negl}(n)$ denotes a negligible function.
\end{definition}

\begin{definition}[Pseudorandom Function (PRF)]
A function $F: \{0,1\}^k \times \{0,1\}^n \to \{0,1\}^m$ is a pseudorandom function if for a randomly chosen key $k \in \{0,1\}^k$, no PPT distinguisher can distinguish between $F_k(\cdot)$ and a truly random function with non-negligible advantage.
\end{definition}

\begin{definition}[Random Oracle]
A random oracle is a theoretical function $\mathcal{O}: \{0,1\}^* \to \{0,1\}^\infty$ where:
\begin{enumerate}
\item For each new input $x$, the output $\mathcal{O}(x)$ is drawn uniformly at random from $\{0,1\}^\infty$
\item For repeated queries, $\mathcal{O}(x)$ returns the same value (consistency)
\item For distinct inputs $x_1 \neq x_2$, outputs are statistically independent
\end{enumerate}
\end{definition}

The key difference: hash functions produce finite output, random oracles produce infinite output. This difference is everything.

\subsection{Kolmogorov Complexity and Information Content}

The Kolmogorov complexity $K(s)$ of a string $s$ is the length of the shortest program that outputs $s$. This gives us a formal way to measure information content.

\begin{theorem}[Information Content of Pseudorandom Sequences]
Let $S$ be a sequence generated by a deterministic extended-output function with seed $s$. Then:
$$K(S[0..n]) \leq |s| + |\text{algorithm}| + O(\log n)$$
regardless of how large $n$ becomes.
\end{theorem}

This theorem reveals the fundamental deception: any ``infinite'' sequence generated deterministically from finite seed contains only finite information. A truly random sequence of length $n$ has $K(S) \approx n$, but our deterministically generated sequences have $K(S) \approx |s|$ bits regardless of length.

\subsection{The Cycle Problem}

Any deterministic function with finite state must eventually cycle.

\begin{lemma}[Cycle Length Bounds]
For any deterministic extended-output function using a hash function with $b$-bit internal state:
\begin{itemize}
\item Maximum cycle length: $2^b$
\item Expected cycle length: $\sqrt{\pi \cdot 2^b / 2}$ (by birthday paradox~\cite{katz2014introduction})
\item Minimum cycle length: 1 (degenerate case)
\end{itemize}
\end{lemma}

For SHA-256 with 256-bit state, the expected cycle length is approximately $2^{128}$---astronomically large but finite. We'll see concrete implementations of such functions in the following sections.

With these mathematical foundations established, we now turn to implementation. Our first attempt will be to build a true random oracle---and watch it fail spectacularly.

\section{OracleDigest: A Constructive Proof of Impossibility}

Let's build a true random oracle and watch it fail:

\begin{lstlisting}
class OracleDigest(Digest):
    def __init__(self, input, entropy=None):
        if entropy is None:
            entropy = lambda: hashlib.sha256(os.urandom(32)).digest()
        self.entropy = entropy
        self.cache = {}  # Maps index -> random byte
    
    def __getitem__(self, index):
        if index not in self.cache:
            self.cache[index] = self.entropy()[0]
        return self.cache[index]
\end{lstlisting}

This implementation reveals why true random oracles are impossible:

\subsection{The Four Impossibilities}

\begin{enumerate}
\item \textbf{Memory Unboundedness}: Each new index adds to the cache. Access enough indices and memory is exhausted. The oracle dies.

\item \textbf{Non-Serializability}: The cache contains random values generated at runtime. You cannot save an \texttt{OracleDigest} to disk and restore it later. The oracle is ephemeral.

\item \textbf{Non-Reproducibility}: Each instance generates different random values. You cannot replay computations or debug. The oracle is untestable.

\item \textbf{Non-Distributability}: Different machines get different oracles. You cannot share an oracle across systems. The oracle is isolated.
\end{enumerate}

\texttt{OracleDigest} is not an implementation---it's an impossibility proof. It shows us what we cannot have, motivating what we must build instead.

Having proven the impossibility of a true random oracle, we now turn to the elegant solution: a deterministic function that successfully pretends to be random through the power of lazy evaluation.

\section{LazyDigest: The Beautiful Deception}

If we can't have true randomness, we'll fake it using lazy evaluation---a computational strategy where values are computed on-demand rather than eagerly pre-computed:

\begin{lstlisting}
class LazyDigest(Digest):
    def __init__(self, seed, hash_fn=hashlib.sha256):
        self.seed = seed
        self.hash_fn = hash_fn
    
    def __getitem__(self, index):
        h = self.hash_fn()
        h.update(self.seed)
        h.update(str(index).encode('utf-8'))
        return h.digest()[0]
\end{lstlisting}

This is beautiful in its simplicity. We're computing:
$$\text{LazyDigest}[i] = h(\text{seed} \| i)[0]$$

\subsection{The Core Lie}

\texttt{LazyDigest} perpetrates a fundamental deception:
\begin{itemize}
\item \textbf{Appears}: Infinite random sequence
\item \textbf{Actually}: Deterministic function with 256 bits of state
\item \textbf{Information content}: $K(\text{LazyDigest}) = 256$ bits + constant
\item \textbf{Apparent information}: Infinite
\end{itemize}

We're achieving a compression ratio of infinity---representing unbounded data with bounded information.

\subsection{Why the Deception Works}

The deception succeeds because of computational hardness:

\begin{theorem}[Computational Indistinguishability \cite{goldreich2001foundations}]
If $h$ is a secure PRF, then for any PPT distinguisher $\mathcal{D}$:
$$\left|\Pr[\mathcal{D}(\text{LazyDigest}_s) = 1] - \Pr[\mathcal{D}(R) = 1]\right| \leq \text{negl}(n)$$
where $R$ is a truly random sequence and $s$ is chosen uniformly.
\end{theorem}

We're not random---we're computationally hard to distinguish from random. This is a weaker guarantee but sufficient for cryptography.

\subsection{The Inevitable Cycle}

Since \texttt{LazyDigest} has finite state, it must eventually repeat:

\begin{observation}
After at most $2^{256}$ queries, \texttt{LazyDigest} must produce a repeated value, entering a cycle. The expected cycle length is approximately $2^{128}$.
\end{observation}

This seems like a fatal flaw, but $2^{128}$ is approximately $10^{38}$. If you queried one billion indices per second, it would take $10^{21}$ years to expect a cycle---far longer than the age of the universe.

The basic LazyDigest construction is elegant but can be enhanced with more sophisticated techniques to extend cycle length and improve security properties.

\section{Extending the Deception: Advanced Constructions}

To extend the cycle length and improve the deception, we can use more sophisticated constructions.

\subsection{Hierarchical Seeding}

Instead of one seed, use a tree of seeds:

\begin{lstlisting}
class HierarchicalLazyDigest:
    def __getitem__(self, index):
        # Three levels: epoch, chunk, position
        epoch = index // (2**40)
        chunk = (index % (2**40)) // (2**20)
        position = index % (2**20)
        
        # Derive hierarchical seeds
        epoch_seed = h(self.master_seed || "epoch" || epoch)
        chunk_seed = h(epoch_seed || "chunk" || chunk)
        value = h(chunk_seed || position)[0]
        
        return value
\end{lstlisting}

This extends the effective period by structuring the state space hierarchically.

\subsection{Sponge Construction}

The sponge construction (used in SHA-3) provides tunable security:

\begin{lstlisting}
class SpongeLazyDigest:
    def __init__(self, seed, capacity=256):
        self.seed = seed
        self.capacity = capacity  # Bits of internal state
        
    def __getitem__(self, index):
        # Initialize sponge state
        state = self.initialize_sponge(self.seed, self.capacity)
        
        # Absorb index
        state = self.absorb(state, index)
        
        # Squeeze out one byte
        return self.squeeze(state, 1)[0]
\end{lstlisting}

Larger capacity means longer cycles but slower generation.

\subsection{Deterministic Rekeying}

Periodically derive new seeds deterministically:

\begin{lstlisting}
class RekeyingLazyDigest:
    def __getitem__(self, index):
        # Rekey every 2^32 indices
        epoch = index // (2**32)
        local_index = index % (2**32)
        
        # Derive epoch key
        if epoch == 0:
            key = self.seed
        else:
            key = h(self.seed || "rekey" || epoch)
        
        return h(key || local_index)[0]
\end{lstlisting}

This prevents even theoretical cycle detection across epochs.

\subsection{Comparison of Constructions}

\begin{table}[h]
\centering
\caption{Comparison of LazyDigest Constructions}
\label{tab:constructions}
\begin{tabular}{|l|c|c|c|c|}
\hline
\textbf{Construction} & \textbf{Cycle Length} & \textbf{Memory} & \textbf{Computation} & \textbf{Security Property} \\
\hline
Basic LazyDigest & $\leq 2^{256}$ & $O(1)$ & $1 \times h$ & Single hash security \\
\hline
Hierarchical & $\approx 2^{316}$ & $O(\log n)$ & $3 \times h$ & Extended state space \\
\hline
Rekeying & Undetectable & $O(k)$ & $2 \times h$ & Forward security \\
\hline
Sponge & $\approx 2^{c}$ & $O(1)$ & $3 \times h$ & Tunable capacity $c$ \\
\hline
Multi-Hash & $\approx 2^{256} \times m$ & $O(1)$ & $1 \times h_i$ & Algorithm rotation \\
\hline
XOR Composite & $\min(h_i) \times m$ & $O(m)$ & $m \times h$ & All-or-nothing security \\
\hline
Composite & Maximum & $O(k + m)$ & $(3+m) \times h$ & Combined benefits \\
\hline
\end{tabular}
\end{table}

Where $h$ denotes a hash operation, $m$ is the number of hash functions, $k$ is the cache size, and $c$ is the sponge capacity.

Our exploration of various LazyDigest constructions reveals deeper connections to fundamental mathematical concepts. We now examine how these implementations relate to broader questions about computability and representation.

\section{Mathematical Perspectives: Uncomputable Reals and Lazy Evaluation}

Our exploration of LazyDigest reveals profound connections to fundamental questions in mathematics and computation. In this section, we explore how random oracles relate to uncomputable real numbers and how lazy evaluation bridges the finite and infinite.

\subsection{Random Oracles and Uncomputable Reals}

There's a profound connection between random oracles and real numbers. Most real numbers are uncomputable---they cannot be generated by any finite algorithm. Consider:

\begin{itemize}
\item The set of computable reals has measure zero in $\mathbb{R}$
\item Almost all real numbers require infinite information to specify
\item For most reals, the shortest ``program'' is the number itself---its infinite digit sequence
\end{itemize}

A random oracle is the cryptographic analog of an uncomputable real:

\begin{theorem}[Random Oracles as Uncomputable Objects]
A true random oracle $\mathcal{O}$ has the following properties:
\begin{enumerate}
\item $K(\mathcal{O}(x)[0..n]) \approx n$ for any input $x$ (maximal Kolmogorov complexity)
\item No finite program can generate $\mathcal{O}$'s outputs
\item The oracle requires infinite information to specify completely
\end{enumerate}
\end{theorem}

This is why true random oracles are impossible to implement: they're the cryptographic equivalent of trying to store an uncomputable real number. Just as Chaitin's constant $\Omega$ (the halting probability) is a well-defined real number that no algorithm can compute, a random oracle is a well-defined function that no algorithm can implement.

The measure-theoretic perspective is illuminating:
\begin{itemize}
\item \textbf{Computable reals} (like $\pi$, $e$, $\sqrt{2}$): Measure zero, finite programs
\item \textbf{Uncomputable reals} (almost all reals): Measure one, infinite information
\item \textbf{LazyDigest outputs}: Computable, finite program, appears random
\item \textbf{Random oracle outputs}: Uncomputable, infinite information, truly random
\end{itemize}

We're using computable functions to approximate uncomputable ones---the same deception that lets us use floating-point arithmetic to approximate real analysis.

\subsection{Lazy Evaluation and Mathematical Constants}

The concept of representing infinite objects through finite programs is not unusual---it's the foundation of algorithmic information theory. Consider $\pi$: we can compute any digit of $\pi$ using a finite algorithm (such as the Bailey-Borwein-Plouffe formula for hexadecimal digits). This is lazy evaluation in its purest form:

\begin{itemize}
\item \textbf{Finite description}: The algorithm is a few hundred bytes
\item \textbf{Infinite output}: We can compute arbitrarily many digits
\item \textbf{On-demand computation}: We only compute the digits we observe
\item \textbf{Composability}: We can implement comparison operations like $\pi < x$ for any computable $x$
\end{itemize}

This parallels our \texttt{LazyDigest} exactly. Just as we can perform algebra on $\pi$ (at least linear operations like $\pi + e$ or $3\pi$), we can perform operations on lazy digests. The key insight: the shortest program that describes the data \emph{is} often an example of lazy evaluation.

The difference between $\pi$ and \texttt{LazyDigest} is not structural but semantic:
\begin{itemize}
\item $\pi$ has mathematical meaning---its digits encode geometric relationships
\item \texttt{LazyDigest} has cryptographic meaning---its bytes encode computational hardness
\end{itemize}

Both achieve the same compression: finite program $\to$ infinite sequence. But there's a crucial distinction:
\begin{itemize}
\item $\pi$ and \texttt{LazyDigest} are computable---they have finite Kolmogorov complexity
\item Most real numbers and true random oracles are uncomputable---they have infinite Kolmogorov complexity
\end{itemize}

We're in the measure-zero set of computable functions, pretending to be in the measure-one set of uncomputable ones. Linear operations (addition, XOR) naturally preserve laziness, while non-linear operations (square roots, modular exponentiation) may require new algorithms. This explains why our XOR construction in the next section works so elegantly---XOR is the perfect linear operation for combining lazy cryptographic sequences.

Having explored the mathematical foundations of our deception, we now turn to the practical security principles that make these constructions robust in real-world applications.

\section{Security Pillars: Defense in Depth}

Our constructions embody several fundamental security principles that provide robustness beyond simple cryptographic primitives.

\subsection{Algorithm Diversity (Defense in Depth)}

The XOR construction demonstrates the principle of algorithm diversity---using multiple independent algorithms so that breaking the system requires breaking all of them simultaneously:

\begin{lstlisting}
class LazyXorMultiHash:
    def __getitem__(self, index):
        result = 0
        for hash_fn in [sha256, sha512, sha3_256, blake2b]:
            result ^= hash_fn(seed || index)[0]
        return result
\end{lstlisting}

\begin{theorem}[XOR Security Amplification \cite{maurer2004indistinguishability}]
Let $H = \{h_1, h_2, ..., h_n\}$ be independent hash functions. The XOR construction $h_{XOR}(x) = \bigoplus_{i=1}^{n} h_i(x)$ is secure if at least one $h_i$ is secure.
\end{theorem}

This provides:
\begin{itemize}
\item \textbf{Hedging against cryptanalysis}: Future attacks may break SHA-2 but not SHA-3
\item \textbf{Quantum resistance}: Different algorithms have different quantum vulnerabilities
\item \textbf{Implementation diversity}: Bugs in one implementation don't compromise the system
\item \textbf{Computational trade-off}: $n$-fold computation for exponential security gain
\end{itemize}

\subsection{Temporal Isolation (Forward Security)}

The rekeying construction provides forward security---compromise of the current state doesn't reveal past outputs:

\begin{lstlisting}
class RekeyingLazyDigest:
    def __getitem__(self, index):
        epoch = index // rekey_interval
        key = derive_epoch_key(epoch)  # One-way derivation
        return h(key || local_index)[0]
\end{lstlisting}

Even if an attacker obtains the key for epoch $n$, they cannot derive keys for epochs $0..n-1$ due to one-way key derivation.

\subsection{Structural Redundancy (Hierarchical Security)}

The hierarchical construction provides security through structural redundancy:

\begin{lstlisting}
class HierarchicalLazyDigest:
    def __getitem__(self, index):
        epoch_seed = h(master || "epoch" || epoch)
        chunk_seed = h(epoch_seed || "chunk" || chunk)
        value = h(chunk_seed || position)[0]
        return value
\end{lstlisting}

This creates multiple security boundaries:
\begin{itemize}
\item Compromise at position level doesn't reveal chunk seed
\item Compromise at chunk level doesn't reveal epoch seed
\item Compromise at epoch level doesn't reveal master seed
\end{itemize}

\subsection{Capacity Reservation (Sponge Security)}

The sponge construction reserves capacity that never leaves the system:

\begin{lstlisting}
class SpongeLazyDigest:
    def __init__(self, seed, capacity=256):
        # 'capacity' bits never exposed to output
        self.capacity = capacity
\end{lstlisting}

This provides:
\begin{itemize}
\item \textbf{State hiding}: Part of the state is never revealed
\item \textbf{Tunable security}: Larger capacity = stronger security
\item \textbf{Collision resistance}: $2^{capacity/2}$ collision resistance
\end{itemize}

\subsection{Compositional Security}

Multiple security pillars can be combined:

\begin{lstlisting}
class UltraSecureLazyDigest:
    def __getitem__(self, index):
        # Combine all security pillars
        xor_hash = self.xor_multi_hash[index]      # Algorithm diversity
        rekeyed = self.rekeying[index]             # Temporal isolation
        hierarchical = self.hierarchical[index]     # Structural redundancy
        sponge = self.sponge[index]                # Capacity reservation
        
        return xor_hash ^ rekeyed ^ hierarchical ^ sponge
\end{lstlisting}

\subsection{The Security Philosophy}

These pillars reflect a fundamental philosophy: \textbf{assume failure and design for resilience}. We don't trust any single:
\begin{itemize}
\item Algorithm (might be broken)
\item Implementation (might have bugs)
\item Time period (might be compromised)
\item Structural level (might be breached)
\end{itemize}

Instead, we create systems where multiple independent failures are required for total compromise. This is the cryptographic equivalent of the Byzantine Generals Problem---achieving security despite potential failures in components.

With our security principles established, we can now develop a rich algebra of operations on lazy digests, mirroring the algebraic operations on mathematical constants.

\section{An Algebra of Operations}

Just as we can perform arithmetic on $\pi$ without computing all its digits, \texttt{LazyDigest} supports rich algebraic operations that preserve laziness. This isn't coincidence---it's a fundamental property of lazy evaluation. The algebra we develop here mirrors the algebra of computable real numbers, but operating in the space of cryptographic sequences rather than mathematical constants.

\subsection{Lazy Operations (Infinite $\to$ Infinite)}

\begin{lstlisting}
class LazyDigest:
    def xor(self, other):
        """XOR with another lazy digest"""
        return LazyXorDigest(self, other)
    
    def slice(self, start, step):
        """Take every step-th element starting at start"""
        return LazySliceDigest(self, start, step)
    
    def transform(self, f):
        """Apply function f to each element"""
        return LazyTransformDigest(self, f)
\end{lstlisting}

These operations compose infinite sequences without materializing them. Note that these are all linear operations---they can be computed element-wise without global knowledge. Non-linear operations (like finding the median of an infinite sequence or computing modular inverses) would require fundamentally different algorithms or may be impossible to compute lazily.

\subsection{Concrete Operations (Infinite $\to$ Finite)}

\begin{lstlisting}
class LazyDigest:
    def truncate(self, n):
        """Take first n bytes (projection)"""
        return Digest(bytes([self[i] for i in range(n)]))
    
    def sample(self, indices):
        """Extract specific positions"""
        return bytes([self[i] for i in indices])
    
    def fold(self, f, init, n):
        """Reduce first n elements"""
        result = init
        for i in range(n):
            result = f(result, self[i])
        return result
\end{lstlisting}

These operations project infinite sequences to finite values.

Having developed our theoretical framework and operational algebra, we now explore practical applications where the LazyDigest deception provides real value.

\section{Applications: Where the Deception Matters}

\subsection{Deterministic Test Data}

Generate reproducible infinite test streams:

\begin{lstlisting}
def test_stream(test_id):
    """Each test gets its own infinite reproducible stream"""
    seed_str = "test_" + str(test_id)
    seed = hashlib.sha256(seed_str.encode()).digest()
    return LazyDigest(seed)

# Same test always gets same data
stream = test_stream("unit_test_1")
test_data = stream.truncate(1024)  # Get 1KB of test data
\end{lstlisting}

\subsection{Memory-Hard Key Derivation}

Force memory access patterns for password hashing:

\begin{lstlisting}
def memory_hard_kdf(password, salt, memory_cost):
    seed = hashlib.sha256(password + salt).digest()
    lazy = LazyDigest(seed)
    
    # Force random memory accesses
    indices = generate_access_pattern(memory_cost)
    memory_data = lazy.sample(indices)
    
    return hashlib.sha256(memory_data).digest()
\end{lstlisting}

\subsection{Blockchain Proof-of-Work}

Model mining as searching the output space:

\begin{lstlisting}
def find_proof_of_work(block_header, difficulty):
    base_seed = hashlib.sha256(block_header).digest()
    
    for nonce in range(2**32):
        # Create candidate by hashing header + nonce
        h = hashlib.sha256()
        h.update(base_seed)
        h.update(nonce.to_bytes(4, 'big'))
        candidate = h.digest()
        
        if leading_zeros(candidate) >= difficulty:
            return nonce
    
    return None  # No solution found
\end{lstlisting}

\section{The Philosophical Core: Randomness as Computational Hardness}

Our journey from OracleDigest to LazyDigest reveals fundamental questions about the nature of randomness, existence, and computation. This section explores the philosophical implications of our cryptographic deception.

\subsection{The Constructivist Challenge: Do These Objects Even Exist?}

Here's a radical question: Are real numbers even a coherent concept? Many mathematicians---constructivists and finitists---argue they're not. From Kronecker's ``God made the integers, all else is the work of man'' to modern constructive mathematics, there's a compelling case that:

\begin{itemize}
\item Only finite objects that can be explicitly constructed are mathematically coherent
\item The real numbers, as a completed infinity, are an incoherent fantasy
\item Most of classical mathematics deals with objects that cannot exist in any meaningful sense
\end{itemize}

The Church-Turing thesis suggests a precise boundary: \textbf{only objects computable on a Universal Turing Machine have potential existence}. Everything else is mere symbol manipulation without referent.

From this perspective, random oracles and most real numbers share the same ontological status---they're both incoherent:

\begin{theorem}[The Constructivist Equivalence]
Under the constructivist interpretation:
\begin{enumerate}
\item An uncomputable real number cannot be said to exist
\item A true random oracle cannot be said to exist
\item Both are equally incoherent mathematical fantasies
\item Only their computable approximations (like $\pi$ and \texttt{LazyDigest}) have potential existence
\end{enumerate}
\end{theorem}

This isn't just philosophy---it has practical implications. We cannot implement the incoherent:
\begin{itemize}
\item We cannot store a ``random'' real number---only rational approximations
\item We cannot implement a true random oracle---only \texttt{LazyDigest}
\item We cannot compute with actual infinities---only finite approximations
\end{itemize}

\texttt{LazyDigest} thus represents something profound: it's the constructivist's answer to the random oracle. Instead of pretending an incoherent object exists, we replace it with something that does exist---a finite program that generates apparently random output on demand.

The ``deception'' of \texttt{LazyDigest} might be more honest than classical mathematics. We're not claiming to have an infinite random sequence; we're openly admitting we have a finite program that produces one. We're not working with incoherent infinities; we're working with coherent finite processes that can generate unbounded output.

In this light, cryptography is inherently constructivist. We can only implement what we can compute, and what we can compute is precisely what can exist on a UTM. The boundary of cryptographic possibility is the boundary of computational existence.

\subsection{Two Theories of Randomness}

We have two incompatible definitions of randomness:

\begin{enumerate}
\item \textbf{Information-theoretic}: A sequence is random if it has high Kolmogorov complexity
\item \textbf{Computational}: A sequence is random if no efficient algorithm can distinguish it from truly random
\end{enumerate}

\texttt{LazyDigest} fails the first definition catastrophically---it has tiny Kolmogorov complexity. But it passes the second definition perfectly (assuming secure hash functions).

\subsection{The P $\neq$ NP Bet}

Every use of \texttt{LazyDigest} is implicitly betting that P $\neq$ NP:

\begin{theorem}
If P = NP, then \texttt{LazyDigest} can be distinguished from random in polynomial time.
\end{theorem}

\begin{proof}
If P = NP, we can invert the hash function $h$ in polynomial time. Given a sequence $S$, we can check if there exists a seed $s$ such that $S[i] = h(s \| i)[0]$ for all observed $i$. This distinguishes \texttt{LazyDigest} from random.
\end{proof}

We're not generating randomness---we're generating a computationally hard problem and hoping nobody can solve it.

\subsection{The Beautiful Paradox}

\texttt{LazyDigest} embodies a beautiful paradox:

\begin{itemize}
\item \textbf{Simple}: Just hash seed concatenated with index
\item \textbf{Complex}: Output appears completely random
\item \textbf{Finite}: Only 256 bits of true information
\item \textbf{Infinite}: Can generate unbounded output
\item \textbf{Deterministic}: Same seed always gives same sequence
\item \textbf{Unpredictable}: Cannot predict next value without seed
\end{itemize}

This paradox is the heart of modern cryptography: simple deterministic processes that appear complex and random to bounded observers.

\subsection{Entropy, Time's Arrow, and Computational Boundedness}

Here's a radical proposition: entropy and randomness are not objective properties of systems but subjective experiences of computationally bounded observers. Consider Laplace's demon---a hypothetical being with unlimited computational power and perfect information:

\begin{quote}
``An intelligence which could know all forces by which nature is animated... would embrace in the same formula the movements of the greatest bodies of the universe and those of the lightest atom; for it, nothing would be uncertain and the future, as the past, would be present to its eyes.'' ---Pierre-Simon Laplace, 1814
\end{quote}

To Laplace's demon:
\begin{itemize}
\item \texttt{LazyDigest} is trivially non-random---just compute $h(\text{seed} \| i)$
\item A broken egg could be reversed---just invert the dynamics
\item The future and past are equally determined and accessible
\item There is no entropy, only deterministic evolution
\end{itemize}

But we are not Laplace's demons. We are computationally bounded, and this boundedness creates our entire experience of reality:

\begin{theorem}[The Computational Arrow of Time]
For computationally bounded observers:
\begin{enumerate}
\item The past is fixed (we remember it)
\item The future is uncertain (we cannot compute it fast enough)
\item Entropy increases (we cannot invert the dynamics)
\item Randomness exists (we cannot distinguish from true random)
\end{enumerate}
\end{theorem}

This isn't just philosophy---it's the foundation of cryptography. A hash function is essentially a ``fast-forward only'' time machine:
\begin{itemize}
\item \textbf{Forward} (past $\to$ future): $h(x)$ is easy to compute
\item \textbf{Backward} (future $\to$ past): Finding $x$ given $h(x)$ is computationally infeasible
\end{itemize}

\subsubsection{The Broken Egg Principle}

A broken egg is a perfect physical analog of a cryptographic hash:
\begin{itemize}
\item \textbf{Breaking} (forward): Easy, just drop it
\item \textbf{Unbreaking} (backward): Computationally infeasible for bounded beings
\item \textbf{Information preserved}: In principle, all information about the original egg exists in the broken pieces
\item \textbf{Practical irreversibility}: We cannot compute the reverse transformation
\end{itemize}

The second law of thermodynamics might not be a law of physics but a statement about computational complexity. Entropy increases not because information is destroyed but because we cannot compute the inverse transformations needed to decrease it.

\subsubsection{Trapdoor Functions as Temporal Asymmetry}

Trapdoor functions formalize this temporal asymmetry:
$$f: X \to Y \text{ where } f \text{ is easy but } f^{-1} \text{ is hard without secret } k$$

This is precisely the structure of our temporal experience:
\begin{itemize}
\item Moving forward in time: Easy (just wait)
\item Moving backward in time: Computationally infeasible
\item The ``trapdoor'': Would be the demon's unlimited computation
\end{itemize}

\texttt{LazyDigest} creates an artificial arrow of time. Given index $i$, we can easily compute future indices but cannot determine what came before without the seed. We've created a computational universe with its own entropy and irreversibility.

\subsubsection{Randomness as Computational Relativity}

Just as Einstein showed that simultaneity is relative to the observer's reference frame, we're suggesting that randomness is relative to the observer's computational frame:

\begin{definition}[Computational Relativity of Randomness]
A sequence $S$ is random relative to computational class $\mathcal{C}$ if no algorithm in $\mathcal{C}$ can distinguish $S$ from truly random with non-negligible advantage.
\end{definition}

This means:
\begin{itemize}
\item To a classical computer: Quantum states appear random
\item To a bounded computer: \texttt{LazyDigest} appears random
\item To an unbounded computer: Nothing appears random
\item To us: The universe might appear random simply because we're computationally weaker than the process generating it
\end{itemize}

The profound implication: If the universe itself is computational, then ``true'' randomness might not exist at all. What we call randomness might always be computational hardness relative to our bounded perspective. The universe might be entirely deterministic, appearing random only because we're trapped in our computational reference frame, unable to compute fast enough to see the determinism.

In this view, \texttt{LazyDigest} isn't simulating randomness---it's demonstrating what randomness actually is: deterministic processes that exceed our computational grasp.

\subsubsection{The Library of Babel and the Specification Problem}

Borges's Library of Babel contains every possible book---and this reveals a fundamental computational asymmetry that mirrors our universe:

\begin{itemize}
\item \textbf{Generation}: Trivially easy to lazily generate all books (just enumerate all strings)
\item \textbf{Specification}: Impossibly hard to find any particular meaningful book
\item \textbf{Location}: Given a book, nearly impossible to determine its index in the enumeration
\end{itemize}

This is precisely the structure of \texttt{LazyDigest} and perhaps reality itself:

\begin{theorem}[The Generation-Specification Asymmetry]
For many computational processes:
\begin{enumerate}
\item Generation is easy: $O(n)$ to generate $n$ elements lazily
\item Specification is hard: $O(2^n)$ to find a specific element
\item Localization is hard: Given an element, determining its index is computationally infeasible
\end{enumerate}
\end{theorem}

Consider the profound parallel to our existence:
\begin{itemize}
\item The universe might ``lazily generate'' itself through simple rules (like cellular automata or quantum evolution)
\item We find ourselves somewhere in this vast generative process
\item But determining ``where'' we are in the space of all possible universes is computationally impossible
\item We can't even determine our index in time without external reference (when did the universe begin?)
\end{itemize}

\texttt{LazyDigest} embodies this same paradox:
\begin{lstlisting}
# Easy: Generate the millionth byte
byte = lazy_digest[1000000]

# Hard: Find which index produces byte 0x42
index = ? # Computationally infeasible without brute force

# Hard: Given a sequence, find the seed
seed = ? # Requires inverting the hash function
\end{lstlisting}

This suggests a radical reinterpretation of reality: We might live in a lazily evaluated universe where:
\begin{itemize}
\item The ``program'' (laws of physics) is simple
\item The output (reality) appears infinitely complex
\item Our location in the computation is unknowable
\item The past is just the indices we've already computed
\item The future is the indices yet to be evaluated
\end{itemize}

The Library of Babel isn't just a thought experiment---it's the structure of computational reality. Every lazy generative process creates its own Library of Babel, where generation is trivial but specification and localization are computationally prohibitive. The universe might be the ultimate Library of Babel, lazily generating all possible states, with conscious observers like us lost somewhere in its infinite stacks, unable to determine our coordinates in the vast space of possibility.

Just as a reader in Borges's library could never find a specific book without its exact coordinates, we cannot specify our exact position in the universe's computational unfolding. We know we're somewhere in the library of reality, reading the book of our existence, but the catalog is computationally inaccessible. This isn't a limitation of our knowledge---it's a fundamental computational barrier.

\section{Performance Considerations}

\subsection{Time Complexity}

\begin{itemize}
\item Single index access: $O(1)$ hash computations
\item Range of $n$ indices: $O(n)$ hash computations
\item No preprocessing required
\item No state maintained between accesses
\end{itemize}

\subsection{Space Complexity}

\begin{itemize}
\item \texttt{OracleDigest}: $O(k)$ where $k$ = unique indices accessed
\item \texttt{LazyDigest}: $O(1)$ regardless of access pattern
\item Hierarchical variants: $O(\log n)$ for tree depth
\end{itemize}

\subsection{Cache Effects}

\texttt{LazyDigest} has perfect cache locality---each index computation is independent. This enables:
\begin{itemize}
\item Parallel generation of different indices
\item GPU acceleration for bulk generation
\item Distributed computation without coordination
\end{itemize}

\section{Future Directions}

\subsection{Theoretical Extensions}

\begin{enumerate}
\item \textbf{Quantum resistance}: How do quantum computers affect the deception?
\item \textbf{Formal verification}: Prove properties in Coq or Isabelle
\item \textbf{Optimal constructions}: What maximizes cycle length for given entropy?
\end{enumerate}

\subsection{Practical Applications}

\begin{enumerate}
\item \textbf{Verifiable Delay Functions}: Use lazy evaluation for time-lock puzzles
\item \textbf{Succinct arguments}: Prove properties about infinite sequences
\item \textbf{Distributed randomness beacons}: Coordinate without sharing state
\end{enumerate}

\subsection{Philosophical Questions}

The deeper we look, the more profound the questions become:

\begin{enumerate}
\item \textbf{Ontological Questions}:
\begin{itemize}
\item Are real numbers a coherent concept, or just a useful fiction we've collectively agreed to?
\item Is a random oracle any less coherent than the real number system itself?
\item Do mathematical objects exist independently of computation, or only through it?
\end{itemize}

\item \textbf{Computational Boundaries}:
\begin{itemize}
\item Does the Church-Turing thesis define the boundary of mathematical existence?
\item Are UTM-computable objects the only ones with ``potential existence''?
\item Is the set of ``existing'' mathematical objects precisely the computable ones?
\end{itemize}

\item \textbf{The Nature of Randomness}:
\begin{itemize}
\item Is there a meaningful difference between ``true'' randomness and computational randomness?
\item If the universe is computational, is anything truly random?
\item Is randomness an objective property or relative to computational power?
\end{itemize}

\item \textbf{Cryptographic Philosophy}:
\begin{itemize}
\item Does the success of \texttt{LazyDigest} provide evidence for P $\neq$ NP?
\item Is cryptography fundamentally about hiding information or creating computational hardness?
\item If only computable objects exist, what is the ontological status of our cryptographic ``deceptions''?
\end{itemize}

\item \textbf{The Constructivist Program}:
\begin{itemize}
\item Is \texttt{LazyDigest} more ``honest'' than classical real analysis?
\item Should we abandon mathematical objects that cannot be implemented on a UTM?
\item Is applied mathematics inherently constructivist while pure mathematics lives in fantasy?
\end{itemize}
\end{enumerate}

These aren't just academic musings. How we answer them determines:
\begin{itemize}
\item What we consider valid cryptographic primitives
\item How we understand security proofs
\item Whether we believe true randomness is achievable or even coherent
\item The foundations we accept for mathematical reasoning about computation
\end{itemize}

\section{Conclusion: Embracing the Deception}

We began with an impossible goal: implement a random oracle that maps finite inputs to infinite random outputs. We showed this is impossible with \texttt{OracleDigest}, which fails in every dimension---memory, persistence, reproducibility, distribution.

But from this impossibility emerged \texttt{LazyDigest}, a beautiful deception that uses just 256 bits to successfully pretend to be infinite. It's not truly random---it has low Kolmogorov complexity and must eventually cycle. But to any computationally bounded observer, it's indistinguishable from true randomness.

This deception is not a bug or a compromise---it's the fundamental insight that makes modern cryptography possible. We don't need true randomness; we need computational hardness. We don't need infinite information; we need finite information that's hard to compress. We don't need perfection; we need to be good enough to fool bounded adversaries.

\texttt{LazyDigest} teaches us that randomness is in the eye of the beholder. To an information-theoretic observer with unlimited computation, it's trivially non-random. To a computational observer with bounded resources, it's perfectly random. Since we're all computational observers, the deception works.

The beauty lies not in hiding the deception but in understanding it. We know exactly what \texttt{LazyDigest} is---a simple deterministic function. We know exactly what it isn't---a true source of randomness. And we know exactly why it works---because distinguishing it from random is computationally hard.

\subsection{Contributions Summary}

This work makes several concrete contributions:

\begin{itemize}
\item \textbf{Theoretical Framework}: We clarified the distinction between information-theoretic and computational randomness through the lens of lazy evaluation and codata.
\item \textbf{Practical Implementations}: We provided elegant Python implementations of multiple LazyDigest variants, each demonstrating different security properties and trade-offs.
\item \textbf{Security Analysis}: We introduced the XOR construction for algorithm diversity and analyzed how combining multiple hash functions provides defense in depth.
\item \textbf{Pedagogical Value}: We created a teaching tool that makes abstract cryptographic concepts tangible through working code.
\item \textbf{Kolmogorov Complexity Perspective}: We analyzed random oracles through information theory, showing how 256 bits can simulate unbounded entropy through computational hardness.
\end{itemize}

This is the beautiful deception at the heart of cryptography: we build finite automatons that pretend to be infinite, deterministic functions that pretend to be random, simple programs that pretend to be complex. And as long as P $\neq$ NP, the deception holds.

In the end, \texttt{LazyDigest} is more than a data structure or an algorithm. It's a lens through which we can understand the nature of randomness, information, and computation. It shows us that in a computational universe, perception is reality. If something appears random to all observers who matter, then for all practical purposes, it is random.

\subsection{A Constructivist Victory}

Perhaps most profoundly, \texttt{LazyDigest} represents a victory for the constructivist program in mathematics. While classical mathematics struggles with the incoherence of actual infinities and uncomputable reals, we've built something better: a finite, computable object that serves all the practical purposes of an infinite random sequence.

Consider the irony: Classical mathematics claims the real numbers ``exist'' despite most being uncomputable, while simultaneously claiming random oracles ``don't exist'' for the same reason. Cryptography cuts through this confusion with brutal honesty---only what we can compute matters, because only what we can compute can be implemented.

In this light, our ``deception'' is more honest than the classical mathematical edifice:
\begin{itemize}
\item We don't claim to have infinite objects---we have finite programs that generate unbounded output
\item We don't pretend uncomputable things exist---we openly work with computable approximations
\item We don't hide behind axioms about completed infinities---we build finite machines and study their behavior
\end{itemize}

Cryptography thus forces us to be constructivists. Every cryptographic primitive must be implementable on a real computer with finite memory and finite time. The boundary of cryptographic possibility is precisely the boundary of computational constructibility. We cannot encrypt with uncomputable keys, we cannot hash to infinite outputs, we cannot use true random oracles---we can only compute.

The deception is beautiful precisely because it's honest. We're not trying to hide what we're doing---we're celebrating it. We're taking 256 bits and through the alchemy of computational hardness, transforming them into infinity. That's not a limitation; that's magic. But it's \emph{constructive} magic---magic that can be implemented, executed, and verified on a Universal Turing Machine.

In a universe that may itself be computational, \texttt{LazyDigest} isn't approximating some ``true'' random oracle that exists in a Platonic realm. It may be as real as randomness gets. The ``deception'' might not be a deception at all---it might be the only coherent way to understand randomness in a computational cosmos.

\bibliographystyle{plain}
\bibliography{references}

\end{document}